\documentclass[a4paper,notitlepage,11pt]{article}

\usepackage[latin1]{inputenc}
\usepackage[lmargin=2.5cm, rmargin=2.5cm,bottom=2.5cm,top=2.5cm]{geometry}
\usepackage{amsthm,amssymb,amstext,graphicx,amsmath,hyperref,csquotes}

\def\Z{\mathbb {Z}}

\def\N{\mathbb {N}}
\def\Q{\mathbb {Q}}
\newtheorem{theorem}{Theorem}

\newtheorem{corollary}{Corollary}

\newtheorem{lemma}{Lemma}

\newtheorem{definition}{Definition}

\title {On automatic subsets of the Gaussian integers}

\author{
Wieb Bosma \\
Radboud University\\
Heyendaalseweg 135\\
6525 AJ Nijmegen\\
Netherlands \\
\href{mailto:bosma@math.ru.nl}{\tt bosma@math.ru.nl}
\and
Robbert Fokkink and Thijmen Krebs \\
TU Delft\\
Mekelweg 4\\
2628 CD Delft\\
Netherlands \\
\href{mailto:r.j.fokkink@tudelft.nl}{\tt r.j.fokkink@tudelft.nl} \\
}

\begin{document}
\maketitle

\begin{abstract}
Suppose that $a$ and $b$ are multiplicatively independent Gaussian integers, that are
both
of modulus~$\geq \sqrt 5$. We prove that
there exist a $X\subset \mathbb Z[i]$ which is $a$-automatic but not $b$-automatic.
This settles a problem of Allouche,
Cateland, Gilbert, Peitgen, Shallit, and Skordev.
\end{abstract}

\noindent
\section{Introduction}
We assume that the reader is familiar with the theory of automatic
sequences as developed in~\cite{AutSeq}. We briefly recall some relevant notions.
A subset $X\subset \mathbb N$ is \textit{$b$-automatic} for a positive
integer $b>1$ if there exists a deterministic finite automaton (DFA)
that accepts all elements of $X$ and rejects all others. The numbers are fed to the
automaton by input strings $w$ that represent them in the
numeration system with digits $\{0,\ldots,b-1\}$ and base $b>1$.
$X$ is \textit{ultimately periodic} if there exist $p,N>0$ such that for
all $n>N$ we have that $n\in X$ if and only if $n+p\in X$.
Such subsets are $b$-automatic for all~$b$.
Two positive integers $a,b>1$ are \textit{multiplicatively dependent} if $a^r=b^s$ for positive
integers $r,s$. For such $a,b$ the notions of $a$-automaticity and $b$-automaticity coincide.
Suppose that $X$ is not ultimately periodic.
According to Cobham's theorem~\cite{cobham},
if $X$ is $a$-automatic and $b$-automatic
then $a$ and $b$ are multiplicatively dependent.
Over the years, this theorem has been extended to substitutive systems~\cite{durandrigo}
and iterative function systems~\cite{Charlier}.
Extending Cobham's theorem to general numeration systems, however, remains a challenge.
In this short note, we take a modest point of view and look
at what can be said if we replace the natural numbers by the Gaussian integers.

The notion of automatic sets was extended from $\mathbb N$ to arbitrary commutative semirings by Allouche,
Cateland, Gilbert, Peitgen, Shallit, and Skordev in \cite{Alloucheetal}. We state the full definition,
but we only consider the ring of Gaussian integers.
\begin{definition}
Let $R$ be a commutative semiring, let $b\in R$ and let $D\subset R$ be a finite subset that contains $0$.
Then $R$ is called a \emph{$(D,b)$-semiring} if every $r\in R\setminus\{0\}$ has a unique representation
\begin{equation}
r=r_sb^s+\ldots+r_1b+r_0,\ s\in\mathbb N,\ r_j\in D,\ 0\leq j\leq s,\ r_s\not=0.
\end{equation}
\end{definition}
The \emph{digit set} $D$ represents the residue classes modulo $b$, and the residue class
$0\text{ mod } b$ is represented by $0$.
Each element $r\in R$ corresponds to a word $w=r_sr_{s-1}\cdots r_0$ in the free monoid $D^*$. The initial
letter of $w$ is non-zero if $r$ is non-zero, i.e., $w\in (D\setminus\{0\})D^*$.
A subset $X\subset R$ is \emph{$(D,b)$-automatic} if there exists a DFA that accepts
$w$ if and only if it represents an element of $X$.

Hansel and Safer~\cite{Hansel} considered automatic subsets for bases $b=-k+i$, for $k\in\Z_{>0}$, with digit
set $D=\{0,1,\ldots,k^2\}$. They were able
to generalize of Cobham's theorem for these bases, under the assumption that the four exponentials
conjecture is true. In this note we prove that the set of powers $\{b^n\colon n\in\mathbb N\}$
is $a$-automatic if and only if $a$ and $b$ are multiplicatively dependent. We use this result
to answer a question from \cite{Alloucheetal}.  Note that automaticity is defined by means of
a digit set $D$, but that we speak about $a$-automaticity and $b$-automaticity without mentioning~$D$. We will explain below why this is so.

\section{Numeration systems of the Gaussian integers}
\noindent
Not all Gaussian integers $b$ can be used to represent $\mathbb Z[i]$ as a $(D,b)$ ring.
If $b$ is a unit then $Rb$ is equal to $R$ and there is only one residue class. 
In this case, $D$ can only contain $0$ and no other element, so a unit base $b$ is ruled out.
Also, if $b$ is equal to $1+i, 1-i$ or $2$, then there exist no digit set $D$ such
that the Gaussian integers are a $(D,b)$-ring. It is possible to circumvent this technical difficulty,
but this requires some notions of numeration systems that we do not want to go into.
We simply avoid it by assuming from now on that $b$ has modulus $|b|\geq \sqrt 5$ (note that this rules
out the bases $2,-2,1+i,1-i,-1+i,1-i$ which are all multiplicatively dependent).

\begin{theorem}[Davio, Deschamps, and Gossart~\cite{Davioetal}]
For every Gaussian integer $|b|\geq \sqrt 5$ there exists a $D$ such that $\mathbb Z[i]$
is a $(D,b)$-ring. In particular, we may take
\begin{equation}\label{eqD}
D=\left\{d\in \mathbb Z[i]\colon -\frac 12\leq \emph{Re}(d/b)<\frac 12,\ -\frac 12\leq \emph{Im}(d/b)<\frac 12\right\}
\end{equation}
\end{theorem}

The original proof of this result appeared in a technical report, which is often quoted
but not easily accessible. Another proof
can be found in \cite{Krebs}, which is available online. It is immediately clear that all elements of
$D$ represent different residue classes modulo $b$. Pick's theorem~\cite{Pick} implies that $D$ contains all residue
classes. The strenuous part of the proof is checking that the operation $z\to (z-d)/b$ terminates,
where $d$ is the digit that represents $z\text{ mod }b$.

Suppose that $R$ is a $(D,b)$-ring as well as a $(D',b)$-ring. Two digit sets $D$ and $D'$ are \emph{linked} if there exists a finite $E\subset R$ containing $0$ such that
$D+E\subset D'+bE$. It is proved in~\cite{Alloucheetal} that this is an equivalence relation
and that any $(D,b)$-automatic set is a $(D',b)$-automatic set if $D$ and $D'$ are linked.
The following lemma shows that we may suppress mentioning the digit set,
if we consider automatic subsets of the Gaussian integers.

\begin{lemma}\label{lemma1}
Suppose
$\mathbb Z[i]$ is a $(D,b)$-ring and a $(D',b)$ ring for two different digit sets. Then $D$ and $D'$ are linked.
\end{lemma}

\begin{proof}
Suppose that $\Delta=\max\{|d|\colon d\in D\}$ and that $\Delta'=\max\{|d'|\colon d'\in D'\}$.
Let $E$ be the set of all Gaussian integers within radius $\Delta+\Delta'$ of the origin.
For an arbitrary $d\in D$ and $e\in E$, let $d'\in D'$ be equal to $(d+e)\text{ mod }b$. Then
$|(d+e-d')/b|< (2\Delta+2\Delta')/|b|$ and so $(d+e-d')/b\in E$. It follows that
$D$ and $D'$ are linked.
\end{proof}

From now on we speak about $b$-automatic subsets of $\mathbb Z[i]$.
If $D$ is a digit set for $b$ then $D\cup bD\cup \cdots\cup b^{j-1}D$ is a digit set for $b^j$
which produces the same words $w$ to represent the Gaussian integers.
It follows that the notions of $b$-automaticity and $b^j$-automaticity are
equivalent.

Let $D$ be the digit set as in Theorem~\ref{lemma1}. Every $z\in\mathbb Z[i]$
corresponds to a word $w\in (D\setminus\{0\})D^*$. We denote the length of that word by $\ell(z)$.

\begin{lemma}\label{length}
There exists a constant $c>0$ such that $\ell(z)\leq k$ if $|z|\leq c\cdot|b|^k$.
\end{lemma}

\begin{proof}
For $r>0$ define $M(r)=\max\{\ell(z)\colon z\in\mathbb Z[i],\ |z|\leq r\}$, which is a maximum
over a finite set. For given $r$ let $z$ be a Gaussian integer such that $M(r)=\ell(z)$. Let $d_0\in D$
be the last digit in the word $w$ that represents $z$. Then
\[
M(r)=\ell(z)=\ell\left(\frac{z-d_0}b\right)+1\leq M\left(\left|\frac{z-d_0}b\right|\right)+1\leq M\left(\left|\frac{z}b\right|+\left|\frac{d_0}b\right|\right)+1.
\]
Since $d_0/b$ has real and imaginary parts bounded by $\frac 12$ in absolute value, we conclude that
\[
M(r)\leq M\left(\frac{r}{|b|}+1\right)+1.
\]
Iterating this inequality and using  our standing assumption that $|b|\geq\sqrt 5$ we find
\begin{eqnarray*}
M\left(|b|^k\right)&\leq& M\left(|b|^{k-1}+1\right)+1\\&\leq& M\left(|b|^{k-2}+\frac 1{|b|}+1\right)+2\leq\cdots\\&\leq&
M\left(1+\frac 1{|b|^{k-1}}+\frac 1{|b|^{k-2}}+\cdots+1\right)+k-1\\
&\leq& M(3)+k-1
\end{eqnarray*}
It follows that we can specify the constant $c$ as $|b|^{-M(3)}$.
\end{proof}

\section{Automatic subsets of the Gaussian integers}

There are several different proofs of Cobham's theorem on automatic subsets of $\mathbb N$, see~\cite{cobham,durandrigo, rigowaxweiler},
but they all
involve the multiplicative group $G=\{a^mb^n\colon m,n\in\Z\}$ on two generators $a,b\in\mathbb N$.
In particular, what is needed is the topological property that $G$ is a dense subset of $(0,\infty)$ if $a$ and $b$
are multiplicatively independent. We will use a topological property of $G$ as a subset of the complex plane.

\begin{lemma}
Let $a,b\in \Z[i]$ be of modulus $|a|,|b|>1$ and consider
$G=\{a^mb^n\colon m,n\in\Z\}$ as a subset of the complex plane.
Then $1\in G$ is an isolated point if and only if $a$ and $b$ are multiplicatively dependent.
\end{lemma}

\begin{proof}
$G$ is a subgroup of the punctured place $\mathbb C^*=\mathbb C\setminus\{0\}$,
which is a multiplicative topological group.
If $1\in G$ is isolated then every point is isolated since $G$ is a group. Furthermore,
if $g_n$ is a sequence in $G$ which converges to $g\in \mathbb C^*$ that is not in $G$, then
$g_ng_m^{-1}$ converges to $1$ as $n,m\to\infty$, contradicting that $1\in G$ is isolated.
Therefore, if $1\in G$ is isolated, then $G$ is closed and discrete.

Suppose that $a$ and $b$ are multiplicatively dependent. Then
the subgroup $A=\{a^n\colon n\in\Z\}$ has finite index in $G$.
Since $A$ is discrete and $G$ us a finite union of cosets of $A$, we
conclude that $G$ is discrete and that $1\in G$ is isolated.

Now suppose that $1\in G$ is isolated.
Since $G$ is closed and discrete, it intersects the annulus $\left\{z\colon 1\leq |z|\leq |a|\right\}$ in a discrete and
closed subset. By compactness, this intersection is finite. Each residue class
of $G/A$ has a representative in $G\cap R_0$. Therefore $G/A$ is finite and we conclude that $a$ and $b$
are multiplicatively dependent.
\end{proof}

\noindent
If $1$ is not isolated, then there exists a
sequence $a^mb^n$ which converges to $1$ and $|n|,|m|\to \infty$. Since $|a|,|b|>1$ the signs of $n$ and $m$
are opposite. Therefore, there exist $n,m\in\N$ such that $\frac{a^m}{b^n}\to 1$.
In fact, if $1$ is not isolated, then none of the elements of $G$ are isolated. By the same
argument we find that for every $u\in G$ there exist $n,m\in\N$ such that $\frac{a^n}{b^m}\to u$.

\begin{lemma}\label{ab}
Let $a,b$ be multiplicatively independent Gaussian integers that generate the multiplicative
group $G$. Let $D$ be a digit set for $b$ as in equation~\ref{eqD}. Then for every
$u\in G\cap \mathbb Z[i]$ there exist
arbitrarily large $m,n\in\N$ such that $a^m=ub^n+z$ for $z\in\Z[i]$ with $\ell(z)\leq n$.
\end{lemma}

\begin{proof}
Since there exist sequences of natural numbers $m$ and $n$ such that $\frac{a^m}{b^n}$ converges to $u$,
there exist arbitrarily large $m,n\in\N$ such that $\left|\frac{a^m}{b^n}- u\right|<c$, with $c$ as
in Lemma~\ref{length}. By this lemma, $\ell(z)\leq n$ for $z=a^m-ub^n$.
\end{proof}

This lemma should be read as follows: considering the numeration system with base~$b$,
suppose that the word $w$ represents a Gaussian integer in $G$. Then there exist an arbitrarily large power $a^m$ that is
represented by a word $v$ that has $w$ as a prefix (which we denote by $w\sqsubset v $).

\begin{theorem}\label{indep}
The set $\{a^n\colon n\in\N\}$ is not $b$-automatic
if $a,b\in\Z[i]$ are multiplicatively independent.
\end{theorem}

\begin{proof}
We adopt the digit set as in equation~\ref{eqD} and consider $\mathbb Z[i]$
as a $(D,b)$-ring.
Arguing by contradiction,
suppose that $A=(Q,q_0,D,\delta,F)$ is a
deterministic finite automaton DFA that accepts the words in $(D\setminus\{0\})D^*$
that correspond to the set $\{a^n\colon n\in\N\}$, and rejects all others.
Since $1$ is not isolated in $G=\{a^nb^m\colon n,m\in\mathbb Z\}$ there exist arbitrarily large
$p_0$ and $q_0$ such that
\[\left|\frac{a^{q_0}}{b^{p_0}}-1\right|\leq c|b|^{-|Q|}\]
where $c$ is as in Lemma~\ref{length}. Choose $p_0$ and $q_0$ such that $p_0>2|Q|$.
Then $a^{q_0}=b^{p_0}+z_0$ for $|z_0|\leq c|b|^{p_0-|Q|}$, which implies that $\ell(z_0)\leq p_0-|Q|$.
Therefore $a^{q_0}$ is represented by a word $w_0$
of length $p_0+1$ that has prefix $10^{|Q|}$.
Note that $w_0$ is not equal to $10^{p_0}$ since $a$ and $b$ are multiplicatively independent.

By applying the previous lemma we find that there is an arbitrarily large power $a^{q_1}$ that
is represented by $w_1$ that has $w_0$ as a prefix.
Repeating this ad infinitum we find an infinite sequence $w_0 \sqsubset w_1 \sqsubset w_2 \sqsubset\ldots$ of words
representing increasing powers $a^{q_0},a^{q_1},a^{q_2},\ldots$.
The initial word $w_0$ has prefix $10^{|Q|}$, and so all the words have this prefix.
This means that $a^{q_j}=b^{p_j}+z_j$ for $\ell(z_j)\leq p_j-|Q|$.
Since the DFA has $|Q|$ states, if we feed $10^{|Q|}$ to the automaton,
then it visits the same state twice. Hence there exist $0\leq s<t\leq |Q|$ such that
$10^s$ and $10^t$ end up in the same state and so do $10^sw$ and $10^tw$ for any word~$w$.
In particular $10^{|Q|}=10^s0^{|Q|-s}$ and $10^{|Q|+t-s}=10^t0^{|Q|-s}$ end up in the same state.
If we pump a multiple of $t-s$ zeroes into the prefix $10^{|Q|}$ of the word $w_j$,
then the final state in the DFA remains the same.
Since the DFA accepts only powers of $a$, pumping these additional zeroes into the prefix must replace $a^{q_j}$
by a higher power of $a$.

Choose indices $j<k$ such the two words $w_j$ and $w_k$ have the same length modulo $t-s$. Pump the appropriate multiple
of $t-s$ zeroes into $w_j$, so that the pumped up word has the same length as $w_k$. These two words cannot be the same since
$w_0$ is a prefix of $w_k$ but it is not a prefix of the pumped up $w_j$.
We thus obtain two words $u,v$ of equal length, say $p$, representing different powers $a^{q},a^{q'}$, both having prefix $10^{|Q|}$.
In particular $a^q=b^p+z$ and $a^{q'}=b^p+z'$ for some $z$ and $z'$  such that $z\not = z'$
and $\max\{\ell(z),\ell(z')\}\leq p-|Q|$. Now we can repeat the
construction and we can pump multiples of $t-s$ zeroes into the prefixes of $u$ and $v$. This produces an infinite
sequence of powers of $a$ that are equal to $b^{p+(t-s)n}+z$ and $b^{p+(t-s)n}+z'$. The differences of these powers
all solve the equation
\[a^q-a^{q'}=z-z'\]
in which $z-z'\not=0$. Clearly, this is impossible since this
equation has only finitely many solutions. We have reached a contradiction and we
conclude that $\{a^n\colon n\in\N\}$ is not $b$-automatic.
\end{proof}

As an immediate corollary, we can answer a question of \cite{Alloucheetal} whether a set that is $(-1+ki)$-automatic
for $k>1$ is necessarily $n$-automatic for some $n\in\mathbb N$. The answer is negative, since $-1+ki$ and $n$ are
multiplicatively independent for all $n>1$.

\begin{corollary}
For each pair of multiplicatively independent $a,b$ of modulus $>1$ there exists a subset of the Gaussian integers which is $a$-automatic but which
is not $b$-automatic.
\end{corollary}

\begin{proof}
We have restricted our attention to bases of modulus $\geq \sqrt 5$. The remaining Gaussian integers are of the form
$\pm1\pm i$ and $\pm 2$. These remaining bases are all multiplicatively dependent to $2$ and have been analysed in \cite{Alloucheetal}. Sets
that are automatic for such bases are equivalent to $2$-automatic sets.
\end{proof}

Cobham's theorem for $\mathbb N$ states that if $X$ is $a$-automatic and $b$-automatic for multiplicatively
independent $a$ and $b$, then $X$ is $c$-automatic for all natural numbers $c$ (even including $c=1$).
Our final result shows that for the Gaussian integers this statement is not true, which indicates
that some care is required if one
wants to extend Cobham's theorem to the Gaussian integers, assuming that this is possible.

\begin{theorem}\label{nats}
$\Z\subset\Z[i]$ is $b$-automatic if and only if
$b^j\in\N$ for some $j\in\N$.
\end{theorem}

\begin{proof}
We consider the digit set of equation \ref{eqD}.
We remark that if $b=2k+1$ is an odd natural number, then the real digits in $D$
are $\{-k,-k+1,\ldots, k-1, k\}$. They form a well studied numeration system for $\Z$,
 in particular if $b=3$, see~\cite{AutSeq}.
For a base $b\in\mathbb N$ the elements of $\Z$ correspond to words that consist
of real digits only (including the empty word, which represents zero).
So $\Z$ is $b$-automatic if $b\in \N$.
Since we know that the notions of $b$-automaticity and $b^j$ automaticity are equivalent,
we may assume that $b^j\not\in\N$ for all $j\in\N$.

Arguing by contradition,
suppose that $A=(Q,q_0,D,\delta,F)$ is a DFA that accepts all words that
represent the integers.
Choose any natural number $a>1$. Then
$a$ and $b$ are multiplicatively independent and all powers of
$a$ are accepted by our DFA. By Lemma~\ref{ab} there exists a power $a^q$ that
is represented by a word $w$ with prefix $10^{|Q|}$.
As before, there exists a $k>0$ such that if we pump an arbitrary multiple of
$k$ zeroes into the prefix of $w$, then the resulting word again gets accepted by the DFA.
In particular, $a^q=b^p+z$ for $\ell(z)<p$ such that $b^{p+k}+z$ and $b^{p+2k}+z$ are
all real integers.
Taking differences we conclude that $b^{p+k}-b^p$ and $b^{p+2k}-b^{p+k}$ are real integers.
Taking quotients, we conclude that $b^k\in\Q$ and since $b$
it is an algebraic integer, $b^k\in\Z$. Which contradicts our assumption.
\end{proof}

\section{Concluding remarks}

Before Cobham proved his theorem, B\"uchi~\cite{Buchi} proved
that the set $\{a^n\colon n\in\N\}$ is $b$-automatic if and only
if $a$ and $b$ are multiplicatively dependent. Our note thus
extends B\"uchi's result from $\mathbb N$ to $\Z[i]$.
The results in our note represent only a part of the MSc thesis of the third author~\cite{Krebs}.
In that thesis it is also shown how to deal with automaticity for unary bases
or exotic numeration systems.

The topological property
of the group $G$ that we have used is elementary. This should be contrasted to the
four exponentials conjecture, which is used by Hansel and Safer
to show that $G$ is dense in the complex plane
for the special bases $a$ and $b$ that they considered in \cite{Hansel}. Since we have no need for such
a deep conjecture, this gives some
hope that a generalization of Cobham's theorem to the Gaussian integers can be achieved with the same modest means as
Cobham used in~\cite{cobham}.

We end with a question, which is a variation on the question in~\cite{Alloucheetal} that we solved.
Suppose that $X\subset \Z[i]$ is closed under complex conjugation and
that it is $b$-automatic for some $b\in \Z[i]$. Is it true that $X$ is $a$-automatic
for some $a\in\mathbb N$?


\begin{thebibliography}{10}

\bibitem{Alloucheetal} {\sc J.-P. Allouche, E. Cateland, W.J. Gilbert, H.-O. Peitgen, J.O. Shallit, G. Skordev},
{\em Automatic maps in exotic numeration systems}, Theory Comput. Syst. {\bf 30}
(1997), 258--331.

\bibitem{AutSeq} {\sc J.-P. Allouche, J.O. Shallit}, {\em Automatic sequences}, Cambridge University press, 2003.

\bibitem{Buchi} {\sc J.R. B\"uchi},
{\em Weak second-order arithmetic and finite automata},
Z. Math. Logik Grundlagen Math. 6 (1960), 66--92.

\bibitem{Charlier} {\sc \'E. Charlier, J. Leroy, M. Rigo}
{\em An analogue of Cobham's theorem for iterated function systems},
Adv. Math. 280 (2015), 86--120.

\bibitem{cobham} {\sc A. Cobham}, {\em On the base-dependence of sets of
numbers recognizable by finite automata}, Math.~Systems Theory, {\bf 3} (1969),
186--192.

\bibitem{Davioetal} {\sc M. Davio, J.P. Deschamps, and C. Gossart}, {\em Complex arithmetic}, Technical Report
R369, MBLE Research Laboratory, Brussels, Belgium, May 1978.

\bibitem{durandrigo} {\sc F. Durand}, 
{\em Cobham's theorem for
substitutions}, J. Eur. Math. Soc. 13 (2011), 1799--1814.

\bibitem{Hansel} {\sc G. Hansel, T. Safer}, {\em Vers un th\a'{e}or\a`{e}me de Cobham pour les entiers
de Gauss}, Bull. Belg. Math. Soc. Simon Stevin \textbf{10}, vol. 5 (2003), 723--735.

\bibitem{Krebs} {\sc T. Krebs}, {\em Automatic maps on the Gaussian integers}, MSc thesis, TU Delft, 2013.
\href{http://repository.tudelft.nl/islandora/object/uuid\%3A7bec892b-9fc7-4349-bbbf-bbfd4c859f19?collection=education}
{\small{\tt http://repository.tudelft.nl}}

\bibitem{Pick} {\sc G. Pick}, {\em Geometric results on number theory} (Geometrisches zur Zahlenlehre), Sitzungsberichte des Deutschen
Natur\-wissenschaftlich-Medicinischen Vereines f\"ur B\"ohmen \enquote{Lotos} in Prag, nr. 8, 9 (1899), 312--319.
\href{https://archive.org/stream/cbarchive_47270_geometrischeszurzahlenlehre-1906/geometrischeszurzahlenlehre-1906#page/n1/mode/2up}
{\tt http://www.biodiversitylibrary.org/item/50207}

\bibitem{rigowaxweiler} {\sc M. Rigo, L. Waxweiler}, {\em A note on syndeticity,
recognizable sets and Cobham's theorem},
Bull.~European Assoc.~Theor.~Comput.~Sci.~{\bf 88} (2006), 169--173.
\end{thebibliography}
\end{document}